\documentclass[11pt]{amsart}

\newcommand\version{April 5, 2026}

\usepackage{amsmath,amsfonts,amsthm,amssymb,amsxtra}
\usepackage{hyperref}
\usepackage{bbm} 
\usepackage{bm} 
\usepackage{accents} 
\usepackage[normalem]{ulem}

\newtheorem{theorem}{Theorem}[section]

\newtheorem{proposition}[theorem]{Proposition}
\theoremstyle{definition}

\newtheorem{remark}[theorem]{Remark}

\numberwithin{equation}{section}

\newcommand{\C}{\mathbb{C}}

\renewcommand{\epsilon}{\varepsilon}

\newcommand{\N}{\mathbb{N}}

\renewcommand{\phi}{\varphi}
\newcommand{\R}{\mathbb{R}}

\newcommand{\Sph}{\mathbb{S}}

\newcommand{\Z}{\mathbb{Z}}

\DeclareMathOperator{\dom}{dom}

\DeclareMathOperator{\spec}{spec}

\DeclareMathOperator{\Tr}{Tr}
\DeclareMathOperator{\tr}{Tr}

\def\bs{\mathbb{S}}

\def\balpha{\bm{\alpha}}

\newcommand{\me}[1]{\mathrm{e}^{#1}}
\newcommand{\one}{\mathbf{1}}

\begin{document}

\title[Sharp upper bounds for the density of relativistic atoms]{Sharp upper bounds\\ for the density of relativistic atoms:\\ Noninteracting case}

\author[R. L. Frank]{Rupert L. Frank}
\address[Rupert L. Frank]{Mathematisches Institut, Ludwig-Maximilians Universit\"at M\"unchen, The\-resienstr. 39, 80333 M\"unchen, Germany, and Munich Center for Quantum Science and Technology (MCQST), Schellingstr. 4, 80799 M\"unchen, Germany, and Mathematics 253-37, Caltech, Pasa\-de\-na, CA 91125, USA}
\email{r.frank@lmu.de}

\author[K. Merz]{Konstantin Merz}
\address[Konstantin Merz]{Institute for Theoretical Physics, ETH Zurich, Wolfgang-Pauli-Strasse 27, 8093 Zurich, Switzerland, and Institut f\"ur Analysis und Algebra, Technische Universit\"at Braunschweig, Universit\"atsplatz 2, 38106 Braun\-schweig, Germany}
\email{k.merz@tu-bs.de}


\date{\version}

\begin{abstract}
  We prove an optimal upper bound for the density of electrons of an infinite Bohr atom (no electron-electron interactions) described by the relativistic operators of Chandrasekhar and Dirac. We also consider densities in each  angular momentum channel separately.
\end{abstract}

\dedicatory{Dedicated to Barry Simon on the occasion of his 80th birthday}

\thanks{Partial support through US National Science Foundation grant DMS-1954995 (R.L.F.), as well as through the German Research Foundation through EXC-2111-390814868 (R.L.F.) and TRR 352-Project-ID 470903074 (R.L.F.) is acknowledged.}

\maketitle


\section{Introduction and main results}

\subsection{The Chandrasekhar--Coulomb operator}

In \cite{Franketal2020P}, together with Heinz Siedentop and Barry Simon, we proved the analogue of Lieb's strong Scott conjecture for atoms described by a relativistic model of Chandrasekhar. We will discuss this result and its background in more detail later in this introduction. For now, the relevant point is that it involves the nonlocal operator
\begin{align}
	\label{eq:defchandra}
	C_\kappa := \sqrt{1-\Delta} - 1 - \kappa\, |x|^{-1} \quad \text{in} \ L^2(\R^3) \,,
\end{align}
which describes a hydrogen atom with effective nuclear charge $\kappa$ in the Chandrasekhar model. The operator is well-defined as a selfadjoint, lower semibounded operator for $\kappa\leq 2/\pi$ as a consequence of the sharp Kato inequality, see \eqref{eq:hardy} below. The negative spectrum of $C_\kappa$ is discrete and we can introduce the density
$$
\one_{(-\infty,0)}(C_\kappa)(x,x)
\qquad\text{for}\ x\in\R^3
$$
as the square sum of the normalized eigenfunctions of $C_\kappa$ corresponding to its negative eigenvalues. It is this density that appears in the strong Scott asymptotics in \cite{Franketal2020P}.

Our main result is the following sharp upper bound on this density.

\begin{theorem}\label{main}
	Let $0<\kappa\leq 2/\pi$. Then there is a constant $A_\kappa<\infty$ such that for all $x\in\R^3$,
	$$
	\one_{(-\infty,0)}(C_\kappa)(x,x) \leq A_\kappa \left( |x|^{-2\eta_\kappa} \, \one(|x|\leq 1) + |x|^{-\frac32} \, \one(|x|>1) \right).
	$$
	Here $\eta_\kappa$ is the unique number in $(0,1]$ such that $(1-\eta_\kappa)\tan\frac{\pi\eta_\kappa}{2} = \kappa$.
\end{theorem}

\begin{remark}
  \label{rem:discussionmain}
	(a) The function $[0,1]\ni \eta\mapsto (1-\eta)\tan\frac{\pi\eta}{2}$ is continuous and strictly monotone increasing with value $0$ for $\eta=0$ and limiting value $2/\pi$ as $\eta\to 1$. Therefore, $\eta_\kappa$ in the theorem is well-defined.\\
	(b) The bound in the theorem for $|x|> 1$ and $\kappa<2/\pi$ is from \cite{Franketal2020P}. The new result here is, on the one hand, the large distance bound for $\kappa =2/\pi$ and, on the other hand the bound for $|x|\leq 1$. Indeed, in \cite{Franketal2020P} we could prove for small $x$ only a bound $|x|^{-3/2}$ when $\kappa<(1+\sqrt 2)/4$ and a bound $|x|^{-2\eta_\kappa-\epsilon}$ for any $\epsilon>0$ when $\kappa\geq (1+\sqrt 2)/4$.\\
	(c) Our new bound is best possible for $|x|\leq 1$. Indeed, already the square of the ground state eigenfunction is bounded from below by a constant times $|x|^{-2\eta_\kappa}$, as shown in \cite[Theorem~1.3]{Jakubowskietal2023}. As we already mentioned in \cite{Franketal2020P}, we believe that our bound is also best possible for $|x|\geq 1$. Indeed, the regime of large $|x|$ is essentially nonrelativistic and for the relevant operator $-\frac12\Delta- \kappa |x|^{-1}$ Heilmann and Lieb \cite{HeilmannLieb1995} showed that the density behaves like a constant times $|x|^{-3/2}$. This is essentially a semiclassical limit.\\
	(d) We believe that the bound in Theorem \ref{main} can be strengthened to the existence, finiteness and positivity of the limits
	$$
	\lim_{x\to 0} |x|^{2\eta_\kappa}\one_{(-\infty,0)}(C_\kappa)(x,x)
	\qquad\text{and}\qquad
	\lim_{|x|\to +\infty} |x|^{\frac32}\one_{(-\infty,0)}(C_\kappa)(x,x) \,.
	$$
	Also, we believe that the second limit is the same as the one that Heilmann and Lieb \cite{HeilmannLieb1995} found in the nonrelativistic case. We plan to return to these questions in the future and we anticipate that the order-sharp bounds in the present paper will play an important role in proving the existence of the limits.	
\end{remark}


\subsection{The Dirac--Coulomb operator}

Our second main result concerns the analogue of Theorem \ref{main} for the Dirac--Coulomb operator
\begin{align}
	\label{eq:defdnu}
	D_\nu = -i\balpha\cdot\nabla + \beta -\nu \, |x|^{-1}\one_{\C^4}
	\quad \text{in}\ L^2(\R^3:\C^4) \,.
\end{align}
Here, $\balpha\cdot\nabla = \sum_{n=1}^3 \alpha_n \partial_n$, where $\alpha_1,\alpha_2,\alpha_3$ and $\beta$ are four Hermitian $\C^{4\times4}$ matrices obeying, with $\alpha_0=\beta$,
$$
\alpha_i \alpha_j + \alpha_j \alpha_i = 2 \delta_{ij} \one_{\C^4}
\qquad\text{for all}\ 0\leq i,j\leq 3 \,.
$$
(There is a traditional choice of these matrices, but any other choice leads to an operator that differs only by a change of basis in $\C^4$.)

Concerning the coupling constant $\nu$ in \eqref{eq:defdnu} we assume $\nu\in(0,1]$. Under this condition the operator \eqref{eq:defdnu} can be defined as a selfadjoint operator (see \cite{Nenciu1976} and also \cite{EstebanLoss2007}) and its spectrum in the interval $(-1,1)$ is discrete. In fact, the eigenvalues are all nonnegative, as follows from the explicit formula \eqref{eq:sommerfeldeigenvalues} below. Our interest is in the density
$$
\tr_{\C^4} \one_{[0,1)}(D_\nu)(x,x)
\qquad\text{for all}\ x\in\R^3 \,,
$$
defined as the square sum of the absolute values of the $L^2(\R^3:\C^4)$-norma\-lized eigenfunctions of $D_\nu$ corresponding to eigenvalues in the interval $[0,1)$. This density appears in the strong Scott asymptotics for the Furry model of an atom, proved in \cite{MerzSiedentop2022}.

Even though the eigenfunctions are explicitly known in terms of Laguerre polynomials, proving a sharp bound on $\tr_{\C^4} \one_{[0,1)}(D_\nu)(x,x)$ is rather subtle.

\begin{theorem}
	\label{main2}
	Let $0<\nu\leq 1$. Then there is a constant $A_\nu<\infty$ such that for all $x\in\R^3$,
	$$
	\tr_{\C^4} \one_{[0,1)}(D_\nu)(x,x) \leq A_\nu \left( |x|^{-2\Sigma_\nu} \one(|x|\leq 1) + |x|^{-\frac 32}\one(|x|>1) \right).
	$$
	Here $\Sigma_\nu := 1-\sqrt{1-\nu^2}$.
\end{theorem}

\begin{remark}
	(a) The bound in the theorem for $|x|>1$ is from \cite{MerzSiedentop2022} when $\nu<1$. The new result here is, on the one hand, the large distance bound for $\nu=1$ and, on the other hand, the bound for $|x|\leq 1$. Indeed, in \cite{MerzSiedentop2022} the authors could prove for small $x$ only a bound $|x|^{-3/2}$ when $\nu<\sqrt{15}/4$ and a bound $|x|^{-2\Sigma_\nu-\epsilon}$ when $\sqrt{15}/4\leq\nu<1$. (We correct here a minor inaccuracy in the statement of \cite[Theorem 3]{MerzSiedentop2022} in the case $\nu=\sqrt{15}/4$.)\\
	(b) Our bound is best possible for $|x|\leq 1$. Indeed, the ground state of $D_\nu$ is $|x|^{-\Sigma_\nu} e^{-\nu|x|}$ times a constant spinor.\\
	(c) As is already mentioned in \cite{MerzSiedentop2022}, we believe that the bound is also best possible for $|x|\geq 1$. The reason is the same as in Theorem \ref{main}, namely that this regime is essentially nonrelativistic.
\end{remark}


\subsection{Background: the strong Scott conjecture}

Let us describe the context in which Theorem \ref{main} is relevant. For more details we refer to the original paper \cite{Franketal2020P}, the alternative proof in \cite{Franketal2023} and the review \cite{Franketal2023T}.

We consider a neutral system of electrons and a fixed nucleus at the origin of charge $Z\in\N$. The electrons interact with each other and with the nucleus via Coulomb forces. In the relativistic model of Chandrasekhar, with $c$ denoting the speed of light, the system is described by the Hamiltonian
\begin{align}
	\label{eq:manybodychandrasekhar}
	\sum_{\nu=1}^Z\left((c^4 - c^2 \Delta_\nu)^{1/2} - c^2 - \frac{Z}{|x_\nu|}\right) + \sum_{1\leq\nu<\mu\leq Z}\frac{1}{|x_\nu-x_\mu|} \quad \text{in} \ \bigwedge_{\nu=1}^Z L^2(\R^3),
\end{align}
where $\bigwedge_{\nu=1}^Z L^2(\R^3)$ denotes the totally antisymmetric tensor product of $Z$ copies of $L^2(\R^3)$ with itself. The choice of this Hilbert space reflects the fact that electrons are fermions, i.e., their wave function is antisymmetric under the exchange of any two electron coordinates. Here, for the sake of simplicity, we ignore the spin of the electrons.

The Hamiltonian \eqref{eq:manybodychandrasekhar} is well-defined as a selfadjoint, lower semibounded operator provided that $Z/c \leq 2/\pi$. Thus, the many electron limit $Z\to\infty$ is necessarily combined with the nonrelativistic limit $c\to\infty$. For simplicity, we shall assume that
$$
\frac{Z}{c} = \kappa \in(0,\tfrac2\pi)
\qquad\text{is fixed} \,.
$$

As shown in \cite{Sorensen2005}, the ground state energy in this limit is given by
$$
E^{\rm TF} \, Z^\frac73 + o(Z^\frac 73) \,,
$$
where $E^{\rm TF}$ is a negative constant, independent of $\kappa$. In fact, $E^{\rm TF} \, Z^\frac73$ is the nonrelativistic Thomas--Fermi energy. This leading order statement is the relativistic analogue of the foundational result by Lieb and Simon \cite{LiebSimon1977}.

Having established the leading order behavior, it is natural to investigate the subleading correction. In the nonrelativistic case, this was predicted by Scott \cite{Scott1952} and derived rigorously in important works by \cite{Hughes1986,SiedentopWeikard1987O,SiedentopWeikard1989}. In the relativistic case this was achieved in \cite{Franketal2008,Solovejetal2008}; see also \cite{Franketal2009,HandrekSiedentop2015,Fournaisetal2020,MengSiedentop2025} for similar results in related models. In each case, the subleading term is of order $Z^2$. Remarkably, in contrast to the leading order term, the coefficient of the subleading correction depends nontrivially on the parameter $\kappa$.

Besides these results about the ground state energy, also results about the ground state itself, or rather its density, are of interest. We let $\Gamma$ denote a (not necessarily pure) ground state of the Hamiltonian in \eqref{eq:manybodychandrasekhar}. Its density $\rho\in L^1(\R^3)$ is defined by
$$
\tr_{\bigwedge_{\nu=1}^Z L^2(\R^3)} \left(\Gamma \sum_{\nu=1}^Z U(x_\nu)\right) = \int_{\R^3} \rho(x) U(x) \,dx
\qquad\text{for all}\ U\in L^\infty(\R^3) \,.
$$
In our result, we consider a sequence of ground states $\Gamma$, depending on $Z = \kappa c$, along a sequence of $Z$'s diverging to infinity. The existence of ground states in the neutral case is well known \cite{Lewisetal1997}, but there is no reason to expect ground states to be unique. Our result is valid for any choice of ground states (and even for approximated ground states in a certain sense).

The main result of our paper \cite{Franketal2020P} with Siedentop and Simon was the convergence
\begin{equation}
	\label{eq:strongscott}
	\lim_{Z\to\infty} \int_{\R^3} c^{-3} \rho(x/c) u(|x|)\,dx = \int_{\R^3} \one_{(-\infty,0)}(C_\kappa)(x,x) u(|x|)\,dx
\end{equation}
for all $\kappa\in(0,2/\pi)$ and a certain class of test functions $u$. In particular, all bounded and compactly supported functions $u$ are admissible.

The meaning of the asymptotics \eqref{eq:strongscott} is that the probability density function of finding one of the $Z$ electrons at a position $x\in\R^3$, when rescaled by $c^{-1}$, converges to the probability density of finding one out of infinitely many electrons at $x$ in a relativistic atom where the electron-electron interactions are not taken into account.

This is a relativistic analogue of the strong Scott conjecture, which was formulated by Lieb \cite{Lieb1981} in 1981 and proved by Iantchenko, Lieb and Siedentop \cite{Iantchenkoetal1996} in 1996. The strong Scott asymptotics for atoms described by the relativistic Furry operator (the atomic many-particle Dirac--Coulomb operator defined in the totally antisymmetric tensor product of $Z$ copies of the positive spectral subspace of the single-particle Dirac--Coulomb operator $D_\nu$) is proved in \cite{MerzSiedentop2022}.

A key ingredient in our proof of \eqref{eq:strongscott} was the pointwise upper bound on the density $\one_{(-\infty,0)}(C_\kappa)(x,x)$. We emphasize that the proof of this bound is much more delicate than in the nonrelativistic case. Indeed, in the latter case the density is bounded near the origin, whereas in the nonrelativistic case a singularity emerges with an exponent that depends on the coupling constant $\kappa$. In the present paper we find a direct and robust way to derive this bound using heat kernel analysis and we believe that these ideas could also help to simplify other parts of the proof of \cite{Franketal2020P} and possibly to extend it to $\kappa=2/\pi$.


\subsection*{Structure of this paper}
	In the next section, we prove the small-distance bound in Theorem \ref{main}, even for a more general class of operators. In Section~\ref{sec:angmom} we state and prove corresponding bounds in each angular momentum channel and we prove the large-distance bound in Theorem \ref{main}. Finally, in Section \ref{s:dirac} we prove Theorem \ref{main2} about the Dirac--Coulomb operator.


\subsection*{Acknowledgments}
	This work grew out of discussions we had with Barry Simon and Heinz Siedentop while working on \cite{Franketal2020P} and we would like to take this opportunity to dedicate the present paper to Barry and to wish him a happy birthday. RLF acknowledges countless fruitful interactions with Barry and continued support throughout the years. Thank you!
	
	We are grateful to Heinz Siedentop for his constant encouragement and discussions. KM would like to thank Tomasz Jakubowski, Kamil Kaleta, and Karol Szczypkowski for discussions and, in particular, Krzysztof Bogdan for his hospitality during stays at Politechnika Wroc\l awska, where part of this work was done.


\section{Optimal Small-distance bounds}


\subsection{A generalization}

We are able to prove an optimal small-distance bound for a more general class of operators than \eqref{eq:defchandra}. For any dimension $d\in\N$ and any order $0<\alpha<d$, we consider the operator
\begin{align}
	C_\kappa := (1-\Delta)^{\alpha/2} - 1 - \kappa \, |x|^{-\alpha} \quad \text{in} \ L^2(\R^d) \,.
\end{align}
This is a bounded perturbation of the homogeneous operator
\begin{align}\label{eq:deflkappa}
	L_\kappa := (-\Delta)^{\alpha/2} - \kappa \, |x|^{-\alpha} \quad \text{in} \ L^2(\R^d) \,.
\end{align}
By the sharp Hardy--Kato--Herbst inequality \cite{Hardy1920,Kato1966,Herbst1977} (see also \cite{Kovalenkoetal1981,Yafaev1999,FrankSeiringer2008}), viz.,
\begin{align}
	\label{eq:hardy}
	L_0 \geq \kappa_{\rm c}^{(\alpha)}(d) \ |x|^{-\alpha}
\end{align}
with
\begin{align}
	\label{eq:defgammacellequalzero}
	\kappa_{\rm c}^{(\alpha)}(d) := \frac{2^{\alpha} \Gamma \left((d+\alpha)/4\right)^2}{\Gamma \left((d-\alpha)/4\right)^2} \,,
\end{align}
the operator $L_\kappa$, and hence $C_\kappa$, is bounded from below for all $\kappa\leq\kappa_{\rm c}^{(\alpha)}(d)$. In particular, this allows us to define $C_\kappa$ as the Friedrichs extension of the corresponding quadratic form on $C_c^\infty(\R^d)$ whenever $\kappa\leq\kappa_{\rm c}^{(\alpha)}(d)$.

%

To state our result, we introduce the function
\begin{align}
	\label{eq:defsigmagammal}
	\Phi_d^{(\alpha)}: (-\alpha,d) \to \R \,,
	\quad \eta \mapsto 
	\Phi_{d}^{(\alpha)}(\eta) 
	:= \frac{2^{\alpha} \Gamma\left(\frac{\eta+\alpha}{2}\right) \Gamma \left(\frac{d-\eta}{2}\right)}{\Gamma \left(\frac{\eta}{2}\right) \Gamma \left(\frac{d-\eta-\alpha}{2}\right)} \,.
\end{align}
It is symmetric around $\eta=(d-\alpha)/2$, where it attains its maximal value
\begin{align}
	\label{eq:defgammacell}
	\kappa_{\rm c}^{(\alpha)}(d) = \Phi_{d}^{(\alpha)}\left(\frac{d-\alpha}{2}\right).
\end{align}
Note also that $\Phi_{d}^{(\alpha)}(0)=0$, $\lim_{\eta\to-\alpha}\Phi_d^{(\alpha)}(\eta)=-\infty$ and that the map $\eta\mapsto\Phi_{d}^{(\alpha)}(\eta)$ is monotonously increasing for $\eta\in(-\alpha,(d-\alpha)/2]$; see \cite{Franketal2008H}. Thus, for each $\kappa\leq\kappa_{\rm c}^{(\alpha)}(d)$, there is a unique $-\alpha<\eta \leq (d-\alpha)/2$ such that $\kappa=\Phi_{d}^{(\alpha)}(\eta)$.

In this section we shall prove the following bound.

\begin{theorem}
	\label{boundrho}
	Let $d\in\N$, $\alpha\in(0,2\wedge d)$ and $\eta\in(0,(d-\alpha)/2]$, and set $\kappa :=\Phi_d^{(\alpha)}(\eta)$. Then, for all $x\in\R^d\setminus\{0\}$, 
	\begin{align}
		\label{eq:boundrho}
		\one_{(-\infty,0)}(C_\kappa)(x,x)
		\lesssim_{d,\alpha,\eta}
		\left(1\wedge |x|\right)^{-2\eta}.
	\end{align}
\end{theorem}

Here and in what follows, we use the notation $A\wedge B:=\min\{A,B\}$ for $A,B\in\R$. Moreover, we write $A\lesssim B$ for two non-negative quantities $A,B\geq 0$ to indicate that there is a constant $C>0$ such that $A\leq C B$. If $C=C_\tau$ depends on a parameter $\tau$, we sometimes write $A\lesssim_\tau B$. The notation $A\sim B$ means $A\lesssim B\lesssim A$.

\begin{remark}
	\label{rem:boundrho}
	(a) This theorem, for $d=3$ and $\alpha=1$, gives the new bound for $|x|\leq 1$ in Theorem \ref{main}. In fact, note that $\Phi_d^{(1)}(\eta) = (1-\eta)\tan\frac{\pi\eta}{2}$ and $\kappa_{\rm c}^{(1)}(3)=2/\pi$.\\
	(b) For any $d,\alpha$ and $\eta$ our bound in Theorem \ref{boundrho} is optimal in the regime $|x|\leq 1$. Indeed, Jakubowski, Kaleta, and Szczypkowski \cite[Theorem~1.3]{Jakubowskietal2023} show that the square of the ground state of $C_\kappa$ is bounded from below by a constant times $|x|^{-2\eta}$ for $|x|\leq 1$.\\
	(c) Jakubowski, Kaleta, and Szczypkowski \cite[Theorem~1.1]{Jakubowskietal2023} also show that the square of any eigenfunction corresponding to a negative eigenvalue is bounded from above by a constant times $|x|^{-2\eta}$ for $|x|\leq 1$. We show that this holds even for the square sum of \emph{all} eigenfunctions corresponding to negative eigenvalues.\\
\end{remark}


\subsection{Proof of Theorem \ref{boundrho}}

We divide the proof of Theorem \ref{boundrho} into three steps.

\medskip

\emph{Step 1.} Let $\delta_x$ be the $d$-dimensional delta distribution, centered at $x\in\R^d\setminus\{0\}$, and let $\gamma_\kappa := \one_{(-\infty,0)}(C_{\kappa})$. Then, by the spectral theorem,  
\begin{align}
	\begin{split}
		\one_{(-\infty,0)}(C_{\kappa})(x,x)
		& = \tr(\one_{(-\infty,0)}(C_{\kappa}) \delta_x \one_{(-\infty,0)}(C_{\kappa})) \\
		& \leq \|\one_{(-\infty,0)}(C_{\kappa}) \me{\frac t2 C_{\kappa}}\|^2 \, \tr \me{-\frac t2 C_{\kappa}} \delta_x \me{-\frac t2 C_{\kappa}}\\
		& = \|\one_{(-\infty,0)}(C_{\kappa}) \me{\frac t2 C_{\kappa}}\|^2 \cdot \me{-tC_{\kappa}}(x,x) \\
		& \leq \me{-tC_{\kappa}}(x,x) \,.
	\end{split}
\end{align}
This bound reduces matters to proving a heat kernel bound for $C_\kappa$.

\medskip

\emph{Step 2.} In this step we further reduce matters to proving a heat kernel bound for the homogeneous operator $L_\kappa$, defined in \eqref{eq:deflkappa}.

\medskip

\emph{Step 2a.} We begin with the case $\kappa=0$ and claim that, for all $x,y\in\R^d$,
\begin{align}
	\label{eq:heatrelativisticaux1}
	\me{-tC_0}(x,y) & \leq \me{t} \cdot \me{-tL_0}(x,y) \,.
\end{align}

Indeed, since for every $t>0$ the function $\lambda\mapsto e^{-t\lambda^{\alpha/2}}$ is completely monotone, by Bernstein's theorem there is nonnegative measure $\mu_t$ on $[0,\infty)$ such that
$$
\int_0^\infty e^{-\tau\lambda} \,d\mu_t(\tau) = e^{-t\lambda^{\alpha/2}}
\qquad\text{for all}\ \lambda\geq 0 \,.
$$
Replacing here $\lambda$ by $\lambda+1$, we also see that
$$
\int_0^\infty e^{-\tau\lambda} e^{-\tau} \,d\mu_t(\tau) = e^{-t(\lambda+1)^{\alpha/2}}
\qquad\text{for all}\ \lambda\geq 0 \,.
$$
Since $e^{\tau\Delta}(x,y)\geq 0$ and $e^{-\tau}\leq 1$ it follows that
\begin{align*}
	e^{-t(-\Delta+1)^{\alpha/2}}(x,y) & = \int_0^\infty e^{\tau\Delta}(x,y) e^{-\tau} \,d\mu_t(\tau) \\
	& \leq  \int_0^\infty e^{\tau\Delta}(x,y) \,d\mu_t(\tau) \\
	& = e^{-t(-\Delta)^{\alpha/2}}(x,y) \,.
\end{align*}
This is the claimed inequality.

\medskip

\emph{Step 2b.} We turn to the case $\kappa\in(0,\kappa_c^{(\alpha)}(d)]$ and claim that, for all $x,y\in\R^d$,
\begin{align}
	\label{eq:heatrelativistic1}
	\begin{split}
		\me{-t C_{\kappa}}(x,y) \leq \me{t} \cdot \me{-tL_\kappa}(x,y) \,.
	\end{split}
\end{align}

To prove this, let $\epsilon>0$ and set $W(x):= \kappa \min\{ \epsilon^{-\alpha},|x|^{-\alpha}\}$. Then, according to the bound from Step 1 we have, for any $n\in\N$,
$$
\left( \me{- \frac t{2n} C_{0}} \me{\frac tn W} \me{- \frac t{2n} C_{0}} \right)^n (x,y) \leq \me{t} \cdot \left( \me{- \frac t{2n} L_{0}} \me{\frac tn W} \me{- \frac t{2n} L_{0}} \right)^n (x,y) \,.
$$
We take $0\leq f,g \in L^2(\R^d)$ and integrate this bound against $f(x)g(y)$. Taking the limit $n\to\infty$, we obtain, by Trotter's product formula
$$
\langle f, \me{-t(C_0-W)} g \rangle \leq \me{t} \,\langle f, \me{-t(L_0-W)} g \rangle \,.
$$
By the maximum principle, we have
$$
\langle f, \me{-t(L_0-W)} g \rangle \leq \langle f, \me{-t L_\kappa} g \rangle \,.
$$
Meanwhile, by the monotone convergence theorem for quadratic forms \cite[Theorem S.16]{ReedSimon1972}, we have
$$
\lim_{\epsilon\to 0_+} \langle f, \me{-t(C_0-W)} g \rangle = \langle f, \me{-tC_\kappa} g \rangle \,.
$$
Thus, we obtain
$$
\langle f, \me{-tC_\kappa} g \rangle \leq \me{t} \,\langle f, \me{-t L_\kappa} g \rangle \,.
$$
Since $f$ and $g$ are arbitrary and since we know that the heat kernels are continuous in $(\R^d\setminus\{0\})\times(\{0\}\times\R^d)$ \cite{Bogdanetal2019,Jakubowskietal2023}, we obtain the claimed inequality \eqref{eq:heatrelativistic1}.

\medskip

\emph{Step 3.} It is now easy to complete the proof. We recall that, according to \cite{Bogdanetal2019}, for any $\alpha\in(0,2\wedge d)$, $\eta\in(0,(d-\alpha)/2]$ and $x,y\in\R^d$ we have
\begin{align}
	\label{eq:heathardy1}
	\me{-tL_\kappa}(x,y)
	\!\sim_{d,\alpha,\eta}\! \left(1\wedge\frac{|x|}{t^{1/\alpha}}\right)^{-\eta} \left(1\wedge\frac{|y|}{t^{1/\alpha}}\right)^{-\eta} \cdot \me{-tL_0}(x,y)
\end{align}
where, as before, $\kappa=\Phi_d^{(\alpha)}(\eta)$. Moreover, for the operator $L_0$, we have
\begin{align*}
	\me{-tL_0}(x,y) \sim_{d,\alpha} \frac{t}{(t+|x-y|)^{d+\alpha}}.
\end{align*}
(In passing, we note that similar bounds in the case $\eta\in(-\alpha,0)$ corresponding to $\kappa<0$ are proved in \cite{JakubowskiWang2020,Choetal2020}.) 

Thus, combining \eqref{eq:heathardy1} with the bounds from Steps 1 and 2b, we arrive at
$$
\one_{(-\infty,0)}(C_{\kappa})(x,x) \lesssim_{d,\alpha,\eta} \frac{\me{t}}{t^{d/\alpha}} \cdot \left(1\wedge\frac{|x|}{t^{1/\alpha}}\right)^{-2\eta}.
$$
Choosing $t=1$ completes the proof of Theorem \ref{boundrho}.\qed


\section{Electrons with prescribed angular momentum}\label{sec:angmom}

A fundamental idea in the proof of the Scott correction in \cite{Hughes1986,SiedentopWeikard1987O} was to consider electrons with given, fixed angular momentum. This technique was used in \cite{Iantchenkoetal1996} and \cite{Franketal2020P} to prove the strong Scott conjecture in the nonrelativistic and the relativistic case, respectively. The suboptimal bound on the density $\one_{(-\infty,0)}(C_\kappa)(x,x)$ in \cite{Franketal2020P} (see Remark~\ref{rem:discussionmain}) was also proved by deriving bounds on the operators appearing in the angular momentum decomposition and then summing.

In contrast, our proof of Theorem \ref{boundrho} completely avoided this technique and worked directly on $\R^d$. Nevertheless, it is interesting to discuss the analogue of Theorem \ref{boundrho} in each angular momentum channel and to see the dependence of the singularity at the origin on the angular momentum. This is the content of Theorem \ref{boundrhoell} below. In the second subsection we will use the bounds at fixed angular momentum to prove the large distance bound in Theorem \ref{main}.


\subsection{Optimal small-distance bounds at fixed angular momentum}

Before stating the result, we introduce some notation. Let $L_d=\N_0$ for $d\geq2$ and $L_1=\{0,1\}$, and for $\ell\in L_d$, called \emph{angular momentum}, we let $\mathcal H_\ell$ denote the subspace of $L^2(\bs^{d-1})$ formed by spherical harmonics of degree $\ell$. We have
$$
L^2(\R^d) = \bigoplus_{\ell\in L_d} L^2(\R_+,r^{d-1}\,dr) \otimes\mathcal H_\ell \,.
$$
The spherical symmetry of $C_\kappa$ and $L_\kappa$ allows us to decompose these operators accordingly, yielding radial operators in $L^2(\R_+,r^{d-1}\,dr)$. For technical reasons it is slightly more convenient to consider the resulting radial operators not in $L^2(\R_+,r^{d-1}\,dr)$, but rather in $L^2(\R_+,r^{d_\ell-1}\,dr)$, where
$$
d_\ell := d+ 2\ell \,.
$$
We consider the unitary operator $U_\ell:L^2(\R_+,r^{d-1}\,dr)\to L^2(\R_+,r^{d_\ell-1}\,dr)$ by $f\mapsto (Uf)(r):=r^{-\ell}f(r)$. Then we have
\begin{align}
	C_\kappa & = \bigoplus_{\ell\in L_d}U_\ell^* C_{\kappa,\ell} U_\ell \otimes \one_{\mathcal H_\ell}, \\
	L_\kappa & = \bigoplus_{\ell\in L_d}U_\ell^* L_{\kappa,\ell} U_\ell \otimes \one_{\mathcal H_\ell}
\end{align}
for certain operators $C_{\kappa,\ell}$ and $L_{\kappa,\ell}$ in $L^2(\R_+,r^{d_\ell-1}\,dr)$. Note that $C_{\kappa,\ell}=C_{0,\ell}-\kappa/r^\alpha$ and $L_{\kappa,\ell}=L_{0,\ell}-\kappa/r^\alpha$. We emphasize that this decomposition is valid for all $0<\alpha<d$, although we will later have to impose an additional upper bound on $\alpha$.

Sharp Hardy--Kato--Herbst inequalities for the operator $L_{\kappa,\ell}$ were shown by in~\cite{LeYaouancetal1997} (for $d=3$ and $\alpha=1$) and~\cite{Yafaev1999}; see also \cite{BogdanMerz2024} for corresponding ground state representations. They state that for $d\in\N$, $\ell\in L_d$, and $\alpha\in(0,d_\ell)$ one has
\begin{align}
	\label{eq:hardyl}
	L_{0,\ell} \geq \kappa_{\rm c}^{(\alpha)}(d_\ell) \ r^{-\alpha} \,,
\end{align}
where the sharp constant $\kappa_{\rm c}^{(\alpha)}(d_\ell)$ is defined in \eqref{eq:defgammacell}. These constants satisfy
\begin{align}
  \kappa_{\rm c}^{(\alpha)}(d_{\ell+1})> \kappa_{\rm c}^{(\alpha)}(d_{\ell}) > \ldots > \kappa_{\rm c}^{(\alpha)}(d_{0}) = \kappa_{\rm c}^{(\alpha)}(d) \quad\text{for}\ \ell \in\N,
\end{align}
for $d\geq 2$, as well as $\kappa_{\rm c}^{(\alpha)}(d_1) > \kappa_{\rm c}^{(\alpha)}(d_0)$ when $d=1$. This follows from the fact that $\Phi^{(\alpha)}_d(\eta)> \Phi^{(\alpha)}_{d-1}(\eta)$ for all $d\geq 2$ and $\eta\in(-\alpha,(d-1-\alpha)/2)$ (see \cite[Proposition 1.4]{BogdanMerz2024}), combined with the monotonicity of $\Phi^{(\alpha)}_d$.

As a consequence, the operators $L_{\kappa,\ell}$, which are initially defined only for $\kappa\leq \kappa_{\rm c}^{(\alpha)}(d)$ and $\alpha<d$, can in fact be defined for $\kappa\leq \kappa_{\rm c}^{(\alpha)}(d_\ell)$ and $\alpha<d_\ell$ and are nonnegative. Similarly, the operators $C_{\kappa,\ell}$ are lower semibounded for $\kappa\leq \kappa_{\rm c}^{(\alpha)}(d_\ell)$ and $\alpha<d_\ell$.

The following result is the analogue of Theorem \ref{boundrho} for the operators $C_{\kappa,\ell}$.

\begin{theorem}
	\label{boundrhoell}
	Let $d\in\N$, $\ell\in L_d$, $\alpha\in(0,2\wedge d_\ell)$ and $\eta\in(0,(d_\ell-\alpha)/2]$, and set $\kappa :=\Phi_{d_\ell}^{(\alpha)}(\eta)$. Then, for all $r>0$,
	\begin{align}
		\label{eq:boundrhoell}
		\one_{(-\infty,0)}(C_{\kappa,\ell})(r,r)
		\lesssim_{d_\ell,\alpha,\eta} \left(1\wedge r\right)^{-2\eta}.
	\end{align}
\end{theorem}

\begin{remark}
	(a) The theorem shows that for fixed $\kappa$, $\alpha$ and $d$ the singularity at the origin becomes weaker as $\ell$ increases. This follows from the fact that the solutions $\eta=\eta_\ell$ of $\kappa=\Phi^{(\alpha)}_{d_\ell}(\eta)$ are decreasing with respect to $\ell$ as a consequence of the monotonicity properties of $\Phi^{(\alpha)}_d(\eta)$.\\ 
	(b) For $\ell=0$ and any $d$, $\ell$, $\alpha$ and $\eta$ our bound in Theorem \ref{boundrhoell} is optimal in the regime $r\leq 1$. Indeed, the square of the ground state of $C_{\kappa,\ell}$ is bounded from below by a constant times $r^{-2\eta}$ for $r\leq 1$. This follows from \cite[Theorem 1.3]{Jakubowskietal2023} in view of the fact that the ground state of $C_\kappa$ is radial. We believe that the same optimality statement holds for any $\ell$. This would follow if one could prove lower heat kernel bounds for $\me{-tC_{\ell,\kappa}}$, matching those in \eqref{eq:heatrelativistic2} below.\\
	(c) A drawback of Theorem \ref{boundrhoell} is that the dependence of the implicit constant on $\ell$ is not controlled. Note that
	\begin{align}
		\label{eq:boundrhoellsummation}
		\one_{(-\infty,0)}(C_\kappa)(x,x)
		= \sum_{\ell\in L_d}\frac{\dim\mathcal H_\ell}{|\bs^{d-1}|}\, |x|^{2\ell} \, \one_{(-\infty,0)}(C_{\kappa,\ell})(|x|,|x|)
	\end{align}
	where (see, e.g., \cite[(2.46)]{AveryAvery2018})
	$$
	\dim\mathcal H_\ell=\frac{(d+2\ell-2)(d+\ell-3)!}{\ell!(d-2)!} \,.
	$$
	Our strategy in \cite{Franketal2020P} was to bound each term in \eqref{eq:boundrhoellsummation} and to sum the bounds with respect to $\ell$. The bounds we obtained were weaker for small $r$, but they had an explicit $\ell$-dependence.
\end{remark}

\begin{proof}[Proof of Theorem \ref{boundrhoell}]
	We follow closely the three steps in the proof of Theorem \ref{boundrho}. We note that the assumption $\alpha<2$ guarantees the nonnegativity of all the heat kernels involved.
	
	In the same way as in Step 1 there, we note that
	\begin{align}
		\begin{split}
			\one_{(-\infty,0)}(C_{\kappa,\ell})(r,r)
			\leq  \me{-tC_{\kappa,\ell}}(r,r) \,.
		\end{split}
	\end{align}
	
	Next, in the same way as in Step 2 there, we prove
	\begin{align}
		\label{eq:heatrelativistic2}
		\begin{split}
			\me{-t C_{\kappa,\ell}}(r,s) \leq \me{t} \cdot \me{-tL_{\kappa,\ell}}(r,s).
		\end{split}
	\end{align}
	Indeed, this is first proved for $\kappa=0$ by using the subordination argument and then extended to $\kappa>0$ by Trotter's product formula.
	
	To carry out the analogue of Step 3 there, we need the heat kernel bounds from 	\cite{Bogdanetal2024,BogdanMerz2025H}. They state that for $\ell\in L_d$, $\alpha\in(0,2\wedge d_\ell)$, $\eta\in(-\alpha,(d_\ell-\alpha)/2]$, and $\kappa=\Phi_{d_\ell}^{(\alpha)}(\eta)$ we have
	\begin{align}
		\label{eq:heathardy2}
		\me{-tL_{\kappa,\ell}}(r,s)
		\sim_{d_\ell,\alpha,\eta} \left(1\wedge \frac{r}{t^{1/\alpha}}\right)^{-\eta} \cdot \left(1\wedge \frac{s}{t^{1/\alpha}}\right)^{-\eta} \cdot \me{-tL_{0,\ell}}(r,s) \,.
	\end{align}
	Moreover, for the operator $L_{0,\ell}$, we have
	\begin{align}
		\label{eq:radialheatkernel}
		\begin{split}
			\me{-tL_{0,\ell}}(r,s)
			& \sim_{d_\ell,\alpha} \frac{t}{|r-s|^{1+\alpha}(r+s)^{d_\ell-1} + t^{\frac{1+\alpha}{\alpha}}(t^{\frac1\alpha}+r+s)^{d_\ell-1}}.
		\end{split}
	\end{align}
	Formula~\eqref{eq:radialheatkernel} can be derived using subordination and sharp bounds for the subordination density function, see, e.g., \cite[Theorem~2.1]{BogdanMerz2025}. See also Grzywny--Trojan \cite{GrzywnyTrojan2021} for a streamlined argument yielding the upper bound in \eqref{eq:radialheatkernel}, and also \cite{PensonGorska2010} and \cite[p.~136]{Betancoretal2010}.
	
	Thus, combining the above bounds, we arrive at
	\begin{align}
		\one_{(-\infty,0)}(C_{\kappa,\ell})(r,r)
		\lesssim_{d_\ell,\alpha,\eta} \left(1\wedge\frac{r}{t^{1/\alpha}}\right)^{-2\eta}\cdot\frac{\me{t}}{t^{d_\ell/\alpha}}.
	\end{align}
	Choosing $t=1$ completes the proof of Theorem \ref{boundrhoell}.
\end{proof}


\subsection{Proof of Theorem \ref{main}}

In the remainder of this section we assume that $d=3$ and $\alpha=1$ and we complete the proof of Theorem \ref{main}.

Also, we will change the notation slightly, in order to be consistent with \cite{Franketal2020P}. From now on, we use the notation $C_{\kappa,\ell}$ and $L_{\kappa,\ell}$ for the operators in the \emph{unweighted} space $L^2(\R_+)$ (with measure $dr$), that is, the conjugate of the operators denoted previously by this symbol under the unitary map $L^2(\R_+,r^{d_\ell-1}dr) \ni f\mapsto r^{(d_\ell-1)/2}f \in L^2(\R_+)$.

\begin{theorem}
	\label{chandral}
	Let $0<\kappa\leq 2/\pi$. Then there are constants $A_\kappa,L_\kappa<\infty$ such that for all $\ell\in\N_0$ with $\ell\geq L_\kappa$ and all $r\in\R_+$,
	\begin{align*}
		\one_{(-\infty,0)}(C_{\kappa,\ell})(r,r) & \leq A_\kappa \left( \frac{r}{(\ell+\frac12)^4} \, \one_{r\leq (\ell+\frac12)^{3/2}} + \frac{r^3}{(\ell+\frac12)^7} \, \one_{(\ell+\frac12)^{3/2}< r\leq (\ell+\frac12)^2} \right. \\
		& \quad \qquad \left. + \frac{1}{\ell+\tfrac12} \, \one_{r> (\ell+\frac12)^2} \right).
	\end{align*}
	Moreover, for $\ell\in\N_0$ with $\ell< L_\kappa$ and all $r\in\R_+$, we have
	$$
	\one_{(-\infty,0)}(C_{\kappa,\ell})(r,r) \leq A_\kappa \left( r^{2+2\ell-2\eta_\kappa(\ell)} \, \one_{r\leq 1} + \one_{r>1} \right),
	$$
	where $\eta_\kappa(\ell)$ is the unique number in $(0,1+\ell]$ such that $\Phi_{3+2\ell}^{(1)}(\eta_\kappa(\ell)) = \kappa$. 
\end{theorem}

Before proving this proposition, let us show how it implies the large-distance part of Theorem \ref{main}.

\begin{proof}[Proof of Theorem \ref{main}]
	As remarked after Theorem \ref{boundrho}, the latter theorem implies the bound in Theorem \ref{main} for $|x|\leq 1$. 
	
	For $|x|>1$ we use the analogue of \eqref{eq:boundrhoellsummation}, which in our present convention for the operators $C_{\kappa,\ell}$ reads	
	\begin{align}
		\label{eq:boundrhoellsummation2}
		\one_{(-\infty,0)}(C_\kappa)(x,x)
		= |x|^{-2} \sum_{\ell\in\N_0}\frac{(2\ell+1)}{4\pi}\, \one_{(-\infty,0)}(C_{\kappa,\ell})(|x|,|x|) \,.
	\end{align}
	Inserting the bound from Theorem \ref{chandral}, we find
	\begin{align*}
		\one_{(-\infty,0)}(C_\kappa)(x,x) & \lesssim_\kappa |x|^{-2} + \sum_{\ell+\frac12\geq |x|^{2/3}} \frac{|x|^{-1}}{(\ell+\frac12)^3} \\
		& \qquad + \sum_{|x|^{2/3} < \ell+\frac12 \leq |x|^{1/2}} \frac{|x|}{(\ell+\frac12)^6}
		+ \sum_{\ell+\frac12 > |x|^{1/2}} |x|^{-2} \,.
	\end{align*}
	The first term on the right side comes from angular momenta $\ell<L_\kappa$. It is $\leq |x|^{-3/2}$. Elementary considerations show that the last two sums are each bounded by $|x|^{-3/2}$, which is the desired quantity, while the remaining sum is $\lesssim |x|^{-7/3}\leq |x|^{-3/2}$. This proves the claimed bound.
\end{proof}

We prove Theorem \ref{chandral} separately for large and small $\ell$.

\begin{proof}[Proof of Theorem \ref{chandral} for large $\ell$]
	We set $a_\ell:=a(\ell+1/2)^{-2}$ with a constant $a>0$ to be determined (independently of $\ell$). We write
	$$
	\one_{(-\infty,0)}(C_{\kappa,\ell})(r,r) = \tr A^*B^*CBA
	$$
	with
	\begin{align*}
		A & := (C_{\kappa,\ell}+a_\ell) \, \one_{(-\infty,0)}(C_{\kappa,\ell}) \,,\\
		B &:= (C_{0,\ell}+a_\ell) \, (C_{\kappa,\ell}+a_\ell)^{-1} \,,\\
		C & := (C_{0,\ell}+a_\ell)^{-1} \delta_r (C_{0,\ell}+a_\ell)^{-1} \,.
	\end{align*}
	These operators are well-defined if
	$$
	a_\ell > |\inf\spec C_{\kappa,\ell}| \,.
	$$
	Since the right side is $\lesssim_\kappa (\ell+1/2)^{-2}$ according to \cite{Franketal2008}, this is satisfied for all sufficiently large $a>0$.
	
	We have
	$$
	\one_{(-\infty,0)}(C_{\kappa,\ell})(r,r) \leq \|A\|^2 \|B\|^2 \tr C
	$$
	and we need to bound the terms on the right side.
	
	Clearly, we have
	$$
	\|A\|\leq a_\ell \lesssim (\ell+\tfrac12)^{-2} \,.
	$$
	Below, we will show that, fixing $a$ sufficiently large, we have for all sufficiently large $\ell$
	$$
	\|B\|\lesssim 1 \,.
	$$
	Finally, from \cite[Lemma A.2]{Franketal2020P} we know that
	$$
	\Tr C \lesssim r \one_{r\leq (\ell+\frac12)^{3/2}} + \Big(  \frac{r}{\ell+\frac12} \Big)^3 \one_{(\ell+\frac12)^{3/2}< r\leq (\ell+\frac12)^2} + (\ell+\tfrac12)^3\one_{r> (\ell+\frac12)^2} \,.
	$$
	Combining these bounds, we obtain the bound claimed in the proposition.
	
	Thus, we are left with estimating $\|B\|$ uniformly in $\ell$. For $\ell\in\N_0$, let 
	\begin{equation}
		\label{eq:defpl}
			p_\ell := \sqrt{-\frac{d^2}{dr^2}+\frac{\ell(\ell+1)}{r^2}} \quad \text{in} \ L^2(\R_+) \,.
	\end{equation}
	For any $\psi\in\dom C_{\kappa,\ell}$ we have, by Hardy's inequality,
	\begin{align*}
		\| (C_{0,\ell}+a_\ell)\psi\| & \leq \| (C_{\kappa,\ell}+a_\ell)\psi\| + \kappa \| r^{-1} \psi\| \\
		& \leq \| (C_{\kappa,\ell}+a_\ell)\psi\| + \kappa (\ell+\tfrac12)^{-1} \| p_\ell \psi\| \,.
	\end{align*}
	To bound the last term on the right side, we use the bound
	$$
	P \leq C A^{-1/2} \left( \sqrt{P^2+1} -1 + A \right)
	\qquad\text{for all}\ P\geq 0 \,,\ A\leq 1\,.
	$$
	(This can be proved by considering separately the cases $P\leq A^{1/2}$, $A^{1/2}<P\leq 1$ and $P>1$.) Applying this inequality with $A=a_\ell$ we obtain
	$$
	\kappa (\ell+\tfrac12)^{-1} \| p_\ell \psi\| \leq C \kappa a^{-1/2} \| (C_{0,\ell}+a_\ell)\psi\| \,.
	$$
	By choosing $a= (2C\kappa)^2$ we find that
	$$
	\| (C_{0,\ell}+a_\ell)\psi\| \leq 2 \| (C_{\kappa,\ell}+a_\ell)\psi\| \,,
	$$
	which shows that $\|B\|\leq 2$. The choice $A=(2C\kappa)^2(\ell+1/2)^{-2}$ satisfies $A\leq 1$ provided $\ell$ is sufficiently large.	
\end{proof}

\begin{proof}[Proof of Theorem \ref{chandral} for small $\ell$]
	We give the proof for arbitrary $\ell\in\N_0$, but without controlling the $\ell$-dependence. With a constant $a$ we write
	$$
	\one_{(-\infty,0)}(C_{\kappa,\ell})(r,r) = \tr A^*B^*CBA
	$$
	with
	\begin{align*}
		A & := (C_{\kappa,\ell}+a) \, \one_{(-\infty,0)}(C_{\kappa,\ell}) \,,\\
		B &:= (L_{\kappa,\ell}+a) \, (C_{\kappa,\ell}+a)^{-1} \,,\\
		C & := (L_{\kappa,\ell}+a)^{-1} \delta_r (L_{\kappa,\ell}+a)^{-1} \,.
	\end{align*}
	These operators are well-defined if
	$$
	a> |\inf\spec C_{\kappa,\ell}| \,,
	$$
	which can be achieved by a suitable choice of $a$ (depending on $\kappa$ and $\ell$).
	
	We have
	$$
	\one_{(-\infty,0)}(C_{\kappa,\ell})(r,r) \leq \|A\|^2 \|B\|^2 \tr C
	$$
	and we need to bound the terms on the right side.
	
	Clearly, we have
	$$
	\|A\|\leq a <\infty \,.
	$$
	The operator $B$ is bounded since the difference $L_{\kappa,\ell} - C_{\kappa,\ell}$ is a bounded operator, so the operator domains of $C_{\kappa,\ell}$ and $L_{\kappa,\ell}$ coincide. In the remainder of this proof, we shall show
	\begin{equation}
		\label{eq:lbound}
			\tr C = (L_{\kappa,\ell}+a)^{-2}(r,r) \lesssim_{a,\ell} 1 \,.
	\end{equation}
	Once this is shown, we obtain $\one_{(-\infty,0)}(C_{\kappa,\ell})(r,r)\lesssim_{\ell} 1$. Meanwhile, Theorem \ref{boundrhoell} in the present normalization says that, abbreviating $\eta=\eta_\kappa(\ell)$,
	$$
	\one_{(-\infty,0)}(C_{\kappa,\ell})(r,r)\lesssim_{\ell} r^{2+2\ell} (1\wedge r)^{-2\eta} \,.
	$$
	Combining these two bounds yields the bound $\lesssim \min\{ r^{2+2\ell-2\eta},1\}$. To obtain the bound claimed in Theorem \ref{chandral} we use the fact that $\eta\leq 1+\ell$.
	Indeed, by monotonicity of $\Phi_{3_\ell}^{(1)}$, the latter inequality is equivalent to $\kappa = \Phi_{3_\ell}^{(1)}(\eta) \leq \Phi_{3_\ell}^{(1)}(1+\ell) = \kappa_{\rm c}^{(1)}(3_\ell)$. This is satisfied, since $\kappa \leq 2/\pi = \kappa_{\rm c}^{(1)}(3) \leq \kappa_{\rm c}^{(1)}(3_\ell)$ for all $\ell\in\N_0$.
		
	It remains to prove \eqref{eq:lbound}. By the functional calculus and the heat kernel estimate \eqref{eq:heathardy2}, taking into account that now we are working with the measure $dr$ rather than $r^{3_\ell-1}\,dr$ where $3_\ell =3+2\ell$,
	\begin{align}
		\label{eq:resolventzeroangaux}
		\begin{split}
			& (L_{\kappa,\ell}+a)^{-2}(r,r)
			= \int_0^\infty dt\, t \, \me{-t(L_{\kappa,\ell}+a)}(r,r) \\
			& \quad \sim_\ell r^{3_\ell-1} \int_0^\infty dt\, t\, \left(1\wedge \frac{r}{t}\right)^{-2\eta} \cdot \frac{t}{t^2(t+r)^{3_\ell-1}} \cdot \me{-ta} \\
			& \quad = r \int_0^\infty dt\, \left(1\wedge \frac{1}{t}\right)^{-2\eta} \cdot \frac{1}{(t+1)^{3_\ell-1}} \cdot \me{-a r t}.
		\end{split}
	\end{align}
	Suppose first $a r\leq1$. Then the right side is bounded from above (and below) by constants times
	\begin{align}
		\label{eq:auxint1}
		\begin{split}
			& r \left(\int_0^1 dt + \int_1^{1/(a r)}dt\, t^{2\eta-3_\ell+1}  + \int_{1/(a r)}^\infty dt\, t^{2\eta-3_\ell+1} \me{-a r t} \right) \\
			& \quad \sim r \left( \one_{2\eta-3_\ell+2<0} + \ln\left(1+\frac1{ar}\right) \one_{2\eta-3_\ell+2=0} + (ar)^{-2\eta + 3_\ell-2} \one_{2\eta-3_\ell+2>0} \right) \\
			& \quad \lesssim_a 1 \,.
		\end{split}
	\end{align}
	In the last inequality we used $1-2\eta+3_\ell-2\geq 0$, which is equivalent to the inequality $\eta\leq 1+\ell$ shown above.
	
	On the other hand, for $a r\geq1$ the right side of \eqref{eq:resolventzeroangaux} is bounded from above (and below) by constants times
	\begin{align}
		\label{eq:auxint2}
		\begin{split}
			& r \left(\int_0^{1/(a r)}dt + \int_{1/(a r)}^1 dt\, \me{-a r t} + \int_1^\infty dt\, t^{2\eta - 3_\ell +1} \me{-a r t}\right) \lesssim_{a} 1 \,.
		\end{split}
	\end{align}
	
	Combining the bounds for $ar\leq 1$ and $ar\geq 1$, we arrive at \eqref{eq:lbound}.
\end{proof}


\section{Dirac--Coulomb operator}
\label{s:dirac}

The goal of this section is to prove Theorem \ref{main2}.


\subsection{Reduction to bounds for fixed angular momentum}

The spherical symmetry of $D_\nu$ allows for a partial wave decomposition. The relevant consequence of this decomposition for us is that, if we let
$$
D_{\nu,k} := \begin{pmatrix}
	1-\frac{\nu}{r} & -\frac{d}{dr} - \frac{k}{r} \\ \frac{d}{dr} - \frac{k}{r} & -1 - \frac{\nu}{r}
\end{pmatrix}
\qquad\text{in}\ L^2(\R_+:\C^2) \,,
$$
then we have 
\begin{equation}
	\label{eq:densitydecomp}
	\tr_{\C^4} \one_{[0,1)}(D_\nu)(x,x) = |x|^{-2} \sum_{k\in\Z\setminus\{0\}} \frac{|k|}{2\pi} \tr_{\C^2} \one_{[0,1)}(D_{\nu,k})(|x|,|x|) \,.
\end{equation}
Before sketching a proof of \eqref{eq:densitydecomp}, we will explain its use in the proof of Theorem \ref{main2}. We shall prove the following bounds on the square sum of the normalized eigenfunctions of the operators $D_{\nu,k}$.

\begin{theorem}
	\label{main2k}
	Let $0<\nu\leq 1$. Then there are constants $A_\nu,B_\nu<\infty$ such that for all $k\in\Z\setminus\{0\}$ and $r\in\R_+$,
	$$
	\tr_{\C^2} \one_{[0,1)}(D_{\nu,k})(r,r) \leq A_\nu \min\left\{ \left( \frac{B_\nu}{\gamma_k} \right)^{4\gamma_k} r^{2\gamma_k} , |k|^{-1} \right\}.
	$$
	Here $\gamma_k:=\sqrt{k^2-\nu^2}$.
\end{theorem}

Theorem \ref{main2k} follows immediately from Propositions \ref{rhokdirac} and \ref{rhokdiracinfinity} below.

\begin{proof}[Proof of Theorem \ref{main2} given Theorem \ref{main2k}]
	With a parameter $\epsilon>0$ to be determined (independent of $k$), we use the first quantity in the minimum in Theorem \ref{main2k} for $r\leq \epsilon |k|^2$ and the second one for $r>\epsilon|k|^2$. Therefore, \eqref{eq:densitydecomp} gives
	\begin{align*}
		\tr_{\C^4} \one_{[0,1)}(D_\nu)(x,x) & \lesssim |x|^{-2} \sum_{|k|\geq (|x|/\epsilon)^{1/2}} |k| \left( \frac{B_\nu}{\gamma_k} \right)^{4\gamma_k} |x|^{2\gamma_k} \\
		& \quad + |x|^{-2} \sum_{|k|< (|x|/\epsilon)^{1/2}} 1 \,.
	\end{align*}
	For the second sum we obtain clearly
	$$
	|x|^{-2} \sum_{|k|< (|x|/\epsilon)^{1/2}} 1 \lesssim_\epsilon |x|^{-3/2} \one_{|x|>\epsilon} \,,
	$$
	which is the claimed long distance behavior.
	
	For the first sum we bound
	\begin{align*}
		 & |x|^{-2} \sum_{|k|\geq (|x|/\epsilon)^{1/2}} |k| \left( \frac{B_\nu}{\gamma_k} \right)^{4\gamma_k} |x|^{2\gamma_k} \\
		 & \quad \leq |x|^{-2+2\gamma_1} \sum_{|k|\geq (|x|/\epsilon)^{1/2}} |k| \left( \frac{B_\nu}{\gamma_k} \right)^{4\gamma_k} (\epsilon |k|^2)^{2(\gamma_k-\gamma_1)} \\
		 & \quad = |x|^{-2+2\gamma_1} \epsilon^{-2\gamma_1} \sum_{|k|\geq (|x|/\epsilon)^{1/2}} |k|^{1-4\gamma_1} \left( \frac{\sqrt\epsilon \, B_\nu \, |k|}{\gamma_k} \right)^{4\gamma_k}.
	\end{align*}
	We can fix $\epsilon>0$ so small that $\frac{\sqrt\epsilon \, B_\nu \, |k|}{\gamma_k} \leq \frac12$ for all $k\in\Z\setminus\{0\}$. This implies the convergence of the above series and, in fact,
	$$
	|x|^{-2+2\gamma_1} \epsilon^{-2\gamma_1} \sum_{|k|\geq (|x|/\epsilon)^{1/2}} |k|^{1-4\gamma_1} \left( \frac{\sqrt\epsilon \, B_\nu \, |k|}{\gamma_k} \right)^{4\gamma_k} \lesssim |x|^{-2+2\gamma_1} e^{-c_\nu |x|^{1/2}}
	$$
	for some $c_\nu>0$. Recalling that $-2+2\gamma_1 = -2\Sigma_\nu$, we have arrived at the bound from Theorem \ref{main2}.	
\end{proof}

For the sake of completeness, we briefly sketch a proof of the partial wave decomposition leading to \eqref{eq:densitydecomp}. This is described in detail in many textbooks, including \cite{Thaller1992,BalinskyEvans2011}, and here we mostly follow \cite[Section 8.6]{Thaller2005}. The starting point is the decomposition
\begin{equation}
	\label{eq:decompsphere}
	L^2(\Sph^2:\C^4) = \bigoplus_{k\in\Z\setminus\{0\}} \mathcal K_k \,,
\end{equation}
where $\mathcal K_k:= \ker(K-k)$ and $K:=\beta(2S\cdot L+1)$. Here $L := x\wedge(-i\nabla)$ and $S:=-\frac i4 \balpha\wedge\balpha$. Moreover, $\dim\mathcal K_k = 4|k|$. These facts can easily be obtained from the corresponding ones for $L^2(\Sph^2:\C^2)$, a selfcontained proof of which is provided in \cite[Appendix~B]{Dolbeaultetal2025}.

Since $\beta$ commutes with $K$ and has eigenvalues $\pm1$, we have
$$
\mathcal K_k = \mathcal K_k^+ \oplus \mathcal K_k^-
\qquad\text{with}\quad
\mathcal K_k^\pm := \mathcal K_k \cap \ker(\beta\mp 1) \,.
$$
It is easy to see that the matrix-multiplication operator $\balpha\cdot\omega$ (where $\omega$ denotes the independent variable on $\Sph^2$) is idempotent and maps $\mathcal K_k^+$ unitarily onto $\mathcal K_k^-$.

The decomposition \eqref{eq:decompsphere} implies that, with the unitary operator $U:L^2(\R^3:\C^4)\to L^2(\R^3:\C^4, |x|^{-2}\,dx)$, $\psi\mapsto |x|^{-1}\psi$,
$$
L^2(\R^3:\C^4) = U^* \Big( \bigoplus_{k\in\Z\setminus\{0\}} L^2(\R_+) \otimes \mathcal K_k \Big)U \,.
$$
The above recalled facts imply that for each $k\in\Z\setminus\{0\}$ one has
$$
L^2(\R_+) \otimes \mathcal K_k = L^2(\R_+:\C^2) \otimes \mathcal K_k^+ \,,
$$
where the space on the right is defined as
$$
L^2(\R_+:\C^2) \otimes \mathcal K_k^+ := \mathrm{span}\left\{ f^+ \Phi - i f^- \balpha\cdot\omega \Phi :\ f\in L^2(\R_+:\C^2),\ \Phi\in \mathcal K_k^+ \right\}.
$$
Here we write $f=\begin{pmatrix}
	f^+ \\ f^-
\end{pmatrix}$. The tensor product notation is justified since, for $f,g\in L^2(\R_+:\C^2)$ and $\Phi,\Psi\in\mathcal K_k^+$,
\begin{align*}
	& \langle f^+ \Phi - i f^- \balpha\cdot\omega \Phi, g^+ \Psi - i g^- \balpha\cdot\omega \Psi \rangle_{L^2(\R^3:\C^4,|x|^{-2}dx)} \\
	& = \langle f, g\rangle_{L^2(\R_+:\C^2)} \langle\Phi,\Psi\rangle_{L^2(\Sph^2:\C^4)} \,.
\end{align*}
Thus, we have
$$
L^2(\R^3:\C^4) = U^* \Big( \bigoplus_{k\in\Z\setminus\{0\}} L^2(\R_+:\C^2) \otimes \mathcal K_k^+ \Big)U \,.
$$

Since (see \cite[Eq.~(8.170)]{Thaller2005})
$$
D_\nu = -i(\balpha\cdot\omega) \left( \tfrac{\partial}{\partial r} + r^{-1} - r^{-1} \beta K \right) + \beta - \nu r^{-1} \,,
$$
we see that the Dirac--Coulomb operator is invariant with respect to this decomposition and that, for each $k$, it acts trivially on the $\mathcal K_k^+$ factor. More precisely, for $f\in L^2(\R_+:\C^2)$ and $\Phi\in\mathcal K_k^+$ we have
$$
U^*(f^+ \Phi - i f^- \balpha\cdot\omega \Phi) \in \dom D_\nu
\qquad\text{if and only if}\qquad
f\in \dom D_{\nu,k}
$$
and, in this case,
$$
U D_\nu U^*(f^+ \Phi - i f^- \balpha\cdot\omega \Phi) = (D_{\nu,k} f)^+ \Phi - i (D_{\nu,k}f)^- \balpha\cdot\omega\Phi \,.
$$

Note also that for $f\in L^2(\R_+:\C^2)$ and $\Phi\in\mathcal K_k^+$ we have
$$
|f^+ \Phi - i f^- \balpha\cdot\omega \Phi|^2 = \left( |f^+|^2 + |f^-|^2 \right) |\Phi|^2 = |f|^2 |\Phi|^2 \,,
$$
where we used the fact that $\Phi$ and $\balpha\cdot\omega\Phi$, which belong to different eigenspaces of $\beta$, are pointwise orthogonal. This, together with the above block-diagonali\-zation of $D_\nu$, implies that
$$
\tr_{\C^4} \one_{[0,1)}(D_\nu)(x,x) = |x|^{-2} \!\!\! \!\sum_{k\in\Z\setminus\{0\}} \!\! \left( \tr_{\C^2} \Pi_k^+(\tfrac{x}{|x|},\tfrac{x}{|x|}) \right) \tr_{\C^2} \one_{[0,1)}(D_{\nu,k})(|x|,|x|) \,,
$$
where $\Pi_k^+$ is the projection in $L^2(\Sph:\C^4)$ onto $\mathcal K_k^+$. By rotation invariance, $\Tr_{\C^2}\Pi_k^+(\omega,\omega)$ is independent of $\omega$ and since $\Tr \Pi_k^+ = \dim\mathcal K_k^+=2|k|$ we have
$$
\Tr_{\C^2} \Pi_k^+(\omega,\omega) = \frac{1}{4\pi} \int_{\Sph^2} \Tr_{\C^2} \Pi_k^+(\omega',\omega') \,d\omega' = \frac{\Tr \Pi_k^+}{4\pi} = \frac{2|k|}{4\pi} \,.
$$
Therefore, we arrive at the claimed identity \eqref{eq:densitydecomp}.


\subsection{Bounds for fixed angular momentum. I}

We now prove bounds on $\tr_{\C^2} \one_{[0,1)}(D_{\nu,k})(r,r)$ that are sharp near the origin.

\begin{proposition}
  \label{rhokdirac}
  Let $0<\nu\leq 1$. Then there are constants $A_\nu,B_\nu<\infty$ such that for all $k\in\Z\setminus\{0\}$ and $r\in\R_+$,
  $$
  \tr_{\C^2} \one_{[0,1)}(D_{\nu,k})(r,r) \leq A_\nu \left( \frac{B_\nu}{\gamma_k} \right)^{4\gamma_k} r^{2\gamma_k}.
  $$
\end{proposition}

For any $\nu$ and $k$, the bound in the proposition is best possible for $r\leq 1$. Indeed, the ground state of $D_{\nu,k}$ for $k>0$ is given by $r^{\gamma_k} \me{- (\nu/|k|) r}$ times a constant spinor. The ground state for $k<0$ is slightly more complicated, but shows the same behavior.

Our proof of Proposition \ref{rhokdirac} relies on the explicit expressions for the eigenvalues and eigenfunction of $D_\nu$. These are due to \cite{Darwin1928,Gordon1928,Pidduck1929} and can be found in many textbooks, for instance, in \cite[p.~427]{Thaller2005}. The eigenvalues of $D_{\nu,k}$ are
\begin{align}
  \label{eq:sommerfeldeigenvalues}
  E_{n,k} = \left(1+\frac{\nu^2}{(n+\gamma_k)^2}\right)^{-1/2},
  \quad\text{where}\ n\in\mathcal{N}_k := \begin{cases} \N_0 & \text{if}\ k>0 \,,\\ \N & \text{if}\ k<0 \,, \end{cases}
\end{align}
with corresponding $L^2(\R_+:\C^2)$-normalized eigenfunctions
\begin{align}
  \label{eq:sommerfeldeigenfunctions}
  \begin{split}
    \begin{pmatrix} \psi_{n,k}^+(r) \\ \psi_{n,k}^-(r)\end{pmatrix}
    & = r^{\gamma_k}\me{-p_{n,k}r} \left(\frac{\nu}{p_{n,k}}+k\right) \,\cdot\,  _1F_1\left(-n,2\gamma_k+1,2p_{n,k}r\right) \begin{pmatrix}N_{n,k}^+ \\ -N_{n,k}^-\end{pmatrix} \\
    & \quad - r^{\gamma_k} \me{-p_{n,k}r} \cdot n \cdot \, _1F_1\left(1-n,2\gamma_k+1,2p_{n,k}r\right) \begin{pmatrix}N_{n,k}^+ \\ N_{n,k}^-\end{pmatrix}
  \end{split}
\end{align}
where we abbreviate
\begin{align}
	p_{n,k} & = \sqrt{1-E_{n,k}^2} \qquad \text{and} \\
	\begin{split}
		N_{n,k}^\pm & = \frac{(2p_{n,k})^{\gamma_k+3/2}}{4 \, \Gamma(2\gamma_k+1)} \left(\frac{(1\pm E_{n,k}) \, \Gamma(2\gamma_k+n+1)}{\nu (\nu+k p_{n,k}) \, n!}\right)^{1/2}.
	\end{split}
\end{align}
Moreover, $_1F_1(a,b,z)$ denotes Kummer's confluent hypergeometric function,
$$
_1F_1(a,b,z) := \sum_{j=0}^\infty \frac{a(a+1)\cdots(a+j-1)}{b(b+1)\cdots(b+j-1)} \frac{z^j}{j!} \,.
$$
Note that if $a$ is a negative integer, which is the case in our application, then only finitely many terms in the series are nonzero.

\begin{proof}[Proof of Proposition~\ref{rhokdirac}]
  The explicit solution recalled above implies that
  \begin{align*}
    \tr_{\C^2} \one_{[0,1)}(D_{\nu,k})(r,r) & = r^{2\gamma_k} \sum_{n\in \mathcal N_k} \left( (N_{n,k}^+)^2 f^+_{n,k}(r) + (N_{n,k}^-)^2 f^-_{n,k}(r) \right)
  \end{align*}
  with
  \begin{align*}
    f^\pm_{n,k}(r) := e^{-2p_{n,k}r} & \Big( \Big(\frac{\nu}{p_{n,k}}+k\Big) \,\cdot\, _1F_1\left(-n,2\gamma_k+1,2p_{n,k}r\right) \\ & \quad \mp n\,\cdot\, _1F_1\left(1-n,2\gamma_k+1,2p_{n,k}r\right) \Big)^2 \,.
	\end{align*}
	We now use \cite[(18.5.12)]{NIST:DLMF} to express the hypergeometric functions in terms of Laguerre polynomials as
	\begin{align}
          \label{eq:laguerre1}
          _1F_1\left(-n,2\gamma_k+1,2p_{n,k} r\right) & = \frac{\Gamma(2\gamma_k+1)n!}{\Gamma(n+2\gamma_k+1)} L_n^{(2\gamma_k)}(2p_{n,k} r), \quad n\in\N_0, \\
          \label{eq:laguerre2}
		_1F_1\left(1-n,2\gamma_k+1,2p_{n,k} r\right) & = \frac{\Gamma(2\gamma_k+1)(n-1)!}{\Gamma(n+2\gamma_k)}L_{n-1}^{(2\gamma_k)}(2p_{n,k} r), \quad n\in\N.
	\end{align}
	Recall Szeg{\H{o}}'s inequality for Laguerre polynomials~\cite{Szego1918} (see also \cite[Eq.~(18.14.8)]{NIST:DLMF}), which states that, for all $\beta\geq 0$ and $x\in[0,\infty)$,
	\begin{align}
		|L_n^{(\beta)}(x)| \leq
		\frac{\Gamma(n+\beta+1)}{\Gamma(\beta+1)\Gamma(n+1)}\me{x/2}.
	\end{align}
	
	For the function $f^-_{n,k}$, we obtain immediately
	$$
	f_{n,k}^-(r) \leq \left( \left| \frac{\nu}{p_{n,k}}+k \right| + n \right)^2 \,.
	$$
	Using
	\begin{align}
		\frac{\nu}{p_{n,k}} = \sqrt{n^2+2n\gamma_k+k^2} \,,
	\end{align}
	we obtain, uniformly in $n$ and $k$,
	$$
	f_{n,k}^-(r) \lesssim (n+|k|)^2 \,.
	$$
	
	For the function $f^+_{n,k}$ we use the fact, which follows immediately from the definition of the hypergeometric series, see also \cite[(13.3.3)]{NIST:DLMF},
	\begin{align}
		\begin{split}
			& \left(\frac{\nu}{p_{n,k}}+k\right)\cdot\, _1F_1\left(-n,2\gamma_k+1,2p_{n,k} r\right) - n \cdot\, _1F_1\left(1-n,2\gamma_k+1,2p_{n,k} r\right) \\
			& \quad = \left(\frac{\nu}{p_{n,k}} + k - n - 2\gamma_k \right) \cdot\, _1F_1\left(-n, 2\gamma_k+1, 2p_{n,k} r\right) \\ & \qquad + 2\gamma_k \cdot\, _1F_1\left(-n, 2\gamma_k, 2p_{n,k} r\right).
		\end{split}
	\end{align}
	Now applying Szeg{\H{o}}'s bound yields
	$$
	f_{n,k}^+(r) \leq \left( \left| \frac{\nu}{p_{n,k}} + k - n - 2\gamma_k \right| + 2\gamma_k \right)^2 .
	$$
	Observing that
	$$
	\frac{\nu}{p_{n,k}} - n = n\, \frac{2 \frac{\gamma_k}{n} + \frac{k^2}{n^2}}{\sqrt{1 + 2 \frac{\gamma_k}{n} + \frac{k^2}{n^2}}+1} \lesssim |k| \,,
	$$
	we arrive at the bound
	$$
	f_{n,k}^+(r) \lesssim k^2 \,.
	$$
	
	To summarize, what we have shown so far is that
	\begin{align*}
          & \tr_{\C^2} \one_{[0,1)}(D_{\nu,k})(r,r) \\
          & \lesssim r^{2\gamma_k} \sum_{n\in\mathcal N_k} \left( (N_{n,k}^+)^2 k^2 + (N_{n,k}^-)^2 (n+|k|)^2 \right) \\
          & \lesssim r^{2\gamma_k} \sum_{n\in\mathcal N_k} G_{n,k} \frac{p_{n,k}^3}{\nu+kp_{n,k}} \left( (1+E_{n,k}) k^2 + (1-E_{n,k})(n+|k|)^2 \right),
	\end{align*}
	where
	$$
	G_{n,k} := \frac{(2p_{n,k})^{2\gamma_k} \, \Gamma(n+2\gamma_k+1)}{\Gamma(2\gamma_k+1)^2 \, n!} \,.
	$$
	Using 
	\begin{align*}
		p_{n,k} \leq \frac{\nu}{n+\gamma_k}
		\qquad\text{and}\qquad
		1-E_{n,k} \leq p_{n,k}^2 \,,
	\end{align*}
	we arrive at
	\begin{equation}
		\label{eq:rhodiracproof}
			\tr_{\C^2} \one_{[0,1)}(D_{\nu,k})(r,r) \lesssim r^{2\gamma_k} \sum_{n\in\mathcal N_k} G_{n,k} \frac{k^2}{(\nu+kp_{n,k})(n+|k|)^3} \,.
	\end{equation}
 
  	In the remainder of this proof we will show that
  	\begin{align}
  		\label{eq:lowerboundnupluskp}
  		\nu+kp_{n,k} \geq
  		\begin{cases}
  			\nu & \quad \text{if} \ k>0 \,, \\
  			c_\nu \, \frac{n}{n+|k|} & \quad \text{if} \ k<0 \,,
  		\end{cases}
  	\end{align}
  	and that
  	\begin{equation}
  		\label{eq:upperboundgnk}
  		  	G_{n,k} \leq A_\nu \left( \frac{B_\nu}{\gamma_k}\right)^{4\gamma_k}
  		  	\qquad\text{for all}\ n\in\mathcal N_k \,,\ k\in\Z\setminus\{0\} \,.
  	\end{equation}
  	
  	It is clear that the proposition follows from \eqref{eq:rhodiracproof}, \eqref{eq:lowerboundnupluskp} and \eqref{eq:upperboundgnk}. Indeed, when $k>0$, we can use that, uniformly in $k$,
  	$$
  	\sum_{n\in\N_0} \frac{k^2}{(n+|k|)^3} \lesssim 1 \,.
  	$$
  	When $k<0$ instead, we use
  	$$
 	\sum_{n\in\N} \frac{k^2}{n(n+|k|)^2} \lesssim 1+ \ln |k|
  	$$
  	and we absorb the logarithmic increase by slightly decreasing the constant $B_\nu$. Thus, it remains to prove \eqref{eq:lowerboundnupluskp} and \eqref{eq:upperboundgnk}.

  	The bound \eqref{eq:lowerboundnupluskp} for $k>0$ is clear. For $k<0$, we write
  	\begin{align*}
  		\nu+kp_{n,k}=\nu\left(1-\frac{|k|}{\sqrt{\nu^2+(n+\gamma_k)^2}}\right) = \nu \,\frac{A-|k|}{A}
  	\end{align*}
  	with
  	$$
  	A:=\sqrt{\nu^2+(n+\gamma_k)^2}.
  	$$
  	Since $A^2-k^2=n^2+2n\gamma_k$, we have
  	\begin{align*}
  		A-|k|=\frac{n^2+2n\gamma_k}{A+|k|}.
  	\end{align*}
  	Clearly,
  	\begin{align*}
  		A \leq n+|k| \,, \quad A+|k|\leq2(n+|k|)
  		\quad\text{and}\quad
  		n + 2\gamma_k \geq c_\nu (n+|k|) \,.
  	\end{align*}
  	Therefore
  	\begin{align*}
          \nu+kp_{n,k} = \nu \, \frac{n^2 + 2n\gamma_k}{A(A+|k|)} \geq \nu c_\nu\, \frac{n}{n+|k|} \,,
  	\end{align*}
  	as claimed. This proves \eqref{eq:lowerboundnupluskp}.
  	
  	We now turn to the proof of \eqref{eq:upperboundgnk}. We prove the bound only for $n\in\N$. The remaining case $n=0$ and $k>0$ follows by a minor modification of the argument. By Stirling's approximation, for any $x_0>0$ there is $0<C_0<\infty$ such that
  	\begin{align*}
          C_0^{-1} x^{x+1/2}\me{-x} \leq \Gamma(x+1) \leq C_0 x^{x+1/2}\me{-x} 
          \qquad\text{for all}\ x\geq x_0 \,.
  	\end{align*}
	We apply this with $x_0 = \min\{2\sqrt{1-\nu^2},1\}$ and obtain, using $p_{n,k}\leq\frac{\nu}{n+\gamma_k}$,
  	\begin{align*}
  		G_{n,k}
  		& \lesssim \frac{(2\nu)^{2\gamma_k}}{(n+\gamma_k)^{2\gamma_k}} \cdot \frac{(n+2\gamma_k)^{n+2\gamma_k+1/2} \, \me{-(n+2\gamma_k)}}{n^{n+1/2} \, \me{-n} \, (2\gamma_k)^{4\gamma_k+1} \, \me{-4\gamma_k}} \\
  		& \lesssim \sqrt{\frac{n+2\gamma_k}{n}}\left(1+\frac{2\gamma_k}{n}\right)^n \cdot \frac{(n+2\gamma_k)^{2\gamma_k}}{(n+\gamma_k)^{2\gamma_k}} \cdot \frac{\me{2\gamma_k} (2\nu)^{2\gamma_k}}{(2\gamma_k)^{4\gamma_k+1}}.
  	\end{align*}
  	Using the bounds
  	\begin{align*}
  		\left(1+\frac{2\gamma_k}{n}\right)^n \leq \me{2\gamma_k} \,, \qquad
  		\frac{n+2\gamma_k}{n+\gamma_k} \leq 2 \,, \qquad
  		\sqrt{\frac{n+2\gamma_k}{n}} \leq \sqrt{1+2\gamma_k} \,,
  	\end{align*}
  	we get
  	\begin{align*}
          G_{n,k} \lesssim \frac{\sqrt{1+2\gamma_k}}{2\gamma_k} \cdot
          \left( \frac{\me{2}\nu}{\gamma_k^2} \right)^{2\gamma_k} \,,
  	\end{align*}
  	as claimed in \eqref{eq:upperboundgnk}. This completes the proof of Proposition \ref{rhokdirac}.
\end{proof}


\subsection{Bounds for fixed angular momentum. II}

We now complement the bounds from Proposition \ref{rhokdirac}, which capture the behavior of the density near the origin, with uniform bounds that are better at large distances. While these bounds do not show a decay in $r$, they do show a decay in $|k|$ and this suffices for our proof of Theorem \ref{main2k}.

\begin{proposition}
	\label{rhokdiracinfinity}
	Let $\nu\in(0,1]$. Then there is a constant $A_{\nu}<\infty$ such that for all $k\in\Z\setminus\{0\}$ and $r\in\R_+$,
	$$
	\tr_{\C^2} \one_{[0,1)}(D_{\nu,k})(r,r) \leq A_{\nu} |k|^{-1} \,.
	$$
\end{proposition}

\begin{proof}[Proof of Proposition \ref{rhokdiracinfinity} for $|k|>1$]
	We begin by introducing some notation. We let
	$$
	\ell_k^+ := \begin{cases}
		k-1 & \text{if}\ k> 0 \,,\\
		-k & \text{if}\ k<0 \,,
	\end{cases}
	\qquad
	\ell_k^- := \ell_k^++1 \,.
	$$
	and, recalling the definition of $p_\ell$ from \eqref{eq:defpl}, we let
	$$
	H_{\nu,k}:= \begin{pmatrix}
		\tfrac12 p_{\ell_k^+}^2 - \tfrac\nu r \\ \tfrac12 p_{\ell_k^-}^2 - \tfrac\nu r
	\end{pmatrix}
	\quad\text{in}\ L^2(\R_+:\C^2) \,.
	$$
	Next, let
	$$
	F_{\nu,k}:= \one_{[0,\infty)}(D_{\nu,k}) \left( D_{\nu,k} - 1 \right) \one_{[0,\infty)}(D_{\nu,k}) 
	\quad\text{in}\ L^2(\R_+:\C^2)
	$$
	and 
	$$
	\Gamma_{\nu,k} := \one_{[0,1)}(D_{\nu,k}) \,.
	$$
	
	With a constant $a>0$ to be determined (depending on $\nu$, but not on $k$) we set $a_k:= a|k|^{-2}$ and write
	$$
	\tr_{\C^2} \one_{[0,1)}(D_{\nu,k})(r,r) = \tr A^*B^*C^*DCBA
	$$
	with
	\begin{align}
		\label{eq:abcddirac}
		\begin{split}
			A & = (F_{\nu,k}+a_k)^{1/2}\, \Gamma_{\nu,k} \,, \\
			B & = (H_{\nu,k} + a_k)^{1/2} \, \Gamma_{\nu,k} \, (F_{\nu,k}+a_k)^{-1/2} \,  \\
			C & = (H_{0,k}+ a_k)^{1/2} \, (H_{\nu,k}+ a_k)^{-1/2} \,, \\
			D & = (H_{0,k}+ a_k)^{-1/2} \ \delta_r \ (H_{0,k}+ a_k)^{-1/2} \,.
		\end{split}
	\end{align}	
	These operators are well-defined if
	\begin{align*}
		a_k & > |\inf\spec F_{\nu,k}| = 1- E_{0,k} = 1 - \left( \frac{k^2 - \nu^2}{k^2} \right)^{1/2}, \quad \text{and}\\
		a_k & > |\inf\spec H_{\nu,k}| = \max\left\{\frac{\nu^2}{2(1+\ell_k^+)^2}, \frac{\nu^2}{2(1+\ell_k^-)^2}\right\} \,,
	\end{align*}
	which can be achieved by a suitable choice of $a$. (To get the first inequality for $\nu=1$, we need $|k|\neq 1$.) 
	
	It follows that
	$$
	\tr_{\C^2} \one_{[0,1)}(D_{\nu,k})(r,r) \leq \|A\|^2 \|B\|^2 \|C\|^2 \tr D \,,
	$$
	and it remains to bound the terms on the right side.
	
	Clearly,
	$$
	\|A\| \leq a_k^{1/2} \lesssim |k|^{-1} \,.
	$$
	Moreover, in \cite[Appendix~B]{Franketal2020P}, we have shown that, whenever $a$ is sufficiently large, we have, uniformly in $k$,
	$$
	\|C\|\lesssim1
	$$ 
	and
	\begin{align}
		\tr D  \lesssim \frac{r}{|k|}\one_{r\leq k^2} + |k| \one_{r> k^2} \,.
	\end{align}
	In the remainder of this proof, we will show that for an appropriate choice of $a$, we have, uniformly in $k$,
	$$
	\| B \| \lesssim 1 \,.
	$$
	Once this is shown, we obtain
	$$
	\tr_{\C^2} \one_{[0,1)}(D_{\nu,k})(r,r) \lesssim |k|^{-3} r \, \one_{r\leq k^2} + |k|^{-1} \one_{r> k^2} \,,
	$$
	which implies, in particular, $\tr_{\C^2} \one_{[0,1)}(D_{\nu,k})(r,r) \lesssim |k|^{-1}$ for all $r\in\R_+$ and therefore the proposition in case $|k|>1$.
	
	Thus, we are left with estimating $\|B\|$ uniformly in $k$. Using
	\begin{align*}
		-\frac12\Delta = (-i\balpha\cdot\nabla+\beta-1)
		+ \frac12 (-i\balpha\cdot\nabla+\beta-1)^2,
	\end{align*}
	we find
	\begin{align}
		\label{eq:identity}
		\Gamma_{\nu,k} \left(H_{\nu,k}+a_k\right)\Gamma_{\nu,k}
		= \Gamma_{\nu,k} (F_{\nu,k}+a_k)\Gamma_{\nu,k}
		+ \frac12 \Gamma_{\nu,k} \left(D_{0,k}-1\right)^2 \Gamma_{\nu,k} \,.
	\end{align}
	Here we used the fact that $L^2\Phi =\ell^\pm_k(\ell^\pm_k+1)\Phi$ for $\Phi\in\mathcal H^\pm_k$ (see the discussion around \cite[Eq.~(8.187)]{Thaller2005}) and, consequently, if $f=\begin{pmatrix}f^+ \\ f^- \end{pmatrix} \in L^2(\R_+:\C^2)$ and $\Phi\in\mathcal K_k^+$,
	$$
	U (-\tfrac12\Delta) U^*(f^+\Phi-if^-\balpha\cdot\omega\Phi) = (H_{\nu,k}f)^+\Phi - i(H_{\nu,k}f)^-\balpha\cdot\omega\Phi \,.
	$$
	
	It remains to control the second term on the right side of \eqref{eq:identity}. We have, for any $\epsilon>0$,
	\begin{equation}\label{eq:prooflarger}
		\begin{split}
			(D_{0,k}-1)^2 & = (D_{\nu,k} + \nu r^{-1}-1)^2 \\
			& \leq (1+\epsilon^{-1}) (D_{\nu,k}-1)^2 + (1+\epsilon) \nu^2 r^{-2} \\
			& \leq (1+\epsilon^{-1}) (D_{\nu,k}-1)^2 + 2 (1+\epsilon) \nu^2 (k-1/2)^{-2} H_{0,k} \,.
		\end{split}
        \end{equation}
	Here we used Hardy's inequality $(\ell+1/2)^2 r^{-2} \leq p_\ell^2$ and 
	$$
	\max\{(\ell_k^++1/2)^{-2},(\ell_k^-+1/2)^{-2}\} = (k-1/2)^{-2} \,.
	$$
	
	Next, we note that for any $\epsilon'>0$ we can choose $a$ depending on $\epsilon'$ and $\nu$ such that
	$$
	H_{0,k} \leq (1+\epsilon') (H_{\nu,k}+a_k) \qquad\text{for all}\ k\in\Z\setminus\{0\} \,.
	$$	
	Indeed, this inequality is equivalent to
	$$
	\frac{\epsilon'}{1+\epsilon'} \begin{pmatrix}
		\frac12 p_{\ell_k^+}^2 \\ \frac12 p_{\ell_k^-}^2
	\end{pmatrix} - \frac{\nu}{r} \geq - a_k \,.
	$$
	Since the left side is bounded from below by $-((\epsilon')^{-1}+1) \frac12 \nu^2 \max\{ (\ell_k^++1)^{-2},(\ell_k^-+1)^{-2}\}$, this inequality holds, provided
	$$
	((\epsilon')^{-1}+1) \frac12 \nu^2 \max\{(\ell_k^++1)^{-2},(\ell_k^-+1)^{-2}\} \leq a_k \,.
	$$
	This can be achieved by choosing $a$ large.
	
	If we insert this bound into \eqref{eq:prooflarger} and project along $\Gamma_{\nu,k}$ we arrive at
	\begin{align*}
		\frac12 \Gamma_{\nu,k} \left(D_{0,k}-1\right)^2 \Gamma_{\nu,k}
		& \leq \frac12(1+\epsilon^{-1}) \Gamma_{\nu,k} (D_{\nu,k}-1)^2 \Gamma_{\nu,k} \\
		& \quad + (1+\epsilon)(1+\epsilon') \nu^2 (k-1/2)^{-2} \Gamma_{\nu,k} (H_{\nu,k}+a_k)\Gamma_{\nu,k} \,.
	\end{align*}
	
	Next, we use the elementary inequality
	$$
	\frac{(1-E)^2}{A -1 + E} \leq \frac{(1-E_0)^2}{A-1+E_0}
	\qquad\text{for all}\ 0< E_0 \leq E \leq 1
	\ \text{and}\ A> 1-E_0
	$$
	to deduce that
	$$
	\Gamma_{\nu,k} (D_{\nu,k}-1)^2 \Gamma_{\nu,k} \leq \frac{(1-E_{0,k})^2}{a_k -1+E_{0,k}} \Gamma_{\nu,k} ( F_{\nu,k} +a_k ) \Gamma_{\nu,k}
	$$
	with the lowest eigenvalue $E_{0,k}$ from \eqref{eq:sommerfeldeigenvalues}. Thus, we have shown that
	\begin{equation*}
		\begin{split}
			\frac12 \Gamma_{\nu,k} \left(D_{0,k}-1\right)^2 \Gamma_{\nu,k}
			& \leq \frac12(1+\epsilon^{-1}) \frac{(1-E_{0,k})^2}{a_k -1+E_{0,k}} \Gamma_{\nu,k} ( F_{\nu,k} +a_k ) \Gamma_{\nu,k} \\
			& \quad + (1+\epsilon)(1+\epsilon') \nu^2 (k-1/2)^{-2} \Gamma_{\nu,k} (H_{\nu,k}+a_k)\Gamma_{\nu,k} \,.
		\end{split}
	\end{equation*}
	
	Note that $|k|>1$. Therefore, we have $(k-1/2)^{-2} \leq 4/9$ and we can choose $\epsilon>0$ and $\epsilon'>0$ so small that
	$$
	(1+\epsilon)(1+\epsilon') (4/9) \leq 1/2 \,.
	$$
	It follows that
	$$
	(1+\epsilon)(1+\epsilon') \nu^2 (k-1/2)^{-2} \leq 1/2 \,.
	$$
	
	We also note that
	$$
	\frac12(1+\epsilon^{-1}) \frac{(1-E_{0,k})^2}{a_k -1+E_{0,k}} \lesssim |k|^{-2} \lesssim 1 \,,
	$$
	so that we can rewrite the inequality that we have shown as
	\begin{equation}
		\label{eq:prooflarger2}
		\begin{split}
			\frac12 \Gamma_{\nu,k} \left(D_{0,k}-1\right)^2 \Gamma_{\nu,k}
			& \leq c\, \Gamma_{\nu,k} ( F_{\nu,k} +a_k ) \Gamma_{\nu,k} + \frac12\, \Gamma_{\nu,k} (H_{\nu,k}+a_k)\Gamma_{\nu,k} \,.
		\end{split}
	\end{equation}
	
	Inserting this into \eqref{eq:identity} gives
	$$
	\Gamma_{\nu,k} (H_{\nu,k}+a_k)\Gamma_{\nu,k} \leq 2(1+c)\, \Gamma_{\nu,k} ( F_{\nu,k} +a_k ) \Gamma_{\nu,k}
	$$
	Multiplying by $(F_{\nu,k}+a_k)^{-1/2}$ (which commutes with $\Gamma_{\nu,k}$) from the left and the right, we arrive at the claimed bound $\|B\|^2 \leq 2(1+c)$. This completes the proof.
\end{proof}

\begin{proof}[Proof of Proposition \ref{rhokdiracinfinity} for $|k|=1$]
	In fact, we give the proof for arbitrary $k\in\Z\setminus\{0\}$, but without controlling the $k$-dependence.
	
	We let
	\begin{align}
		\label{eq:relationchandradirac}
		\kappa := \begin{cases} \sqrt{1-\nu^2}\cot\left(\frac{\pi}{2}\sqrt{1-\nu^2}\right) & \text{if}\ \nu<1 \,,\\
			\frac\pi 2  & \text{if}\ \nu=1 \,,
		\end{cases}
	\end{align}
	and introduce the operator
	$$
	\tilde L_{\kappa,k} := \begin{pmatrix}
		p_{\ell_k^+} - \tfrac\kappa r \\ p_{\ell_k^-} - \tfrac\kappa r
	\end{pmatrix}
	\quad\text{in}\ L^2(\R_+:\C^2) \,.
	$$
	
	With a constant $a$ we write
	$$
	\tr_{\C^2} \one_{[0,1)}(D_{\nu,k})(r,r) = \tr A^*B^*CBA
	$$
	with
	\begin{align}
		A & := (F_{\nu,k} + a) \, \one_{[0,1)}(D_{\nu,k}) \,, \\
		B & := (\tilde L_{\kappa,k}+a) \, \one_{[0,\infty)}(D_{\nu,k}) \, (F_{\nu,k} + a)^{-1} \,, \\
		C & := (\tilde L_{\kappa,k}+a)^{-1}\delta_r (\tilde L_{\kappa,k}+a)^{-1} \,,
	\end{align}
	These operators are well-defined if
	\begin{align*}
		a & > |\inf\spec F_{\nu,k}| = 1- E_{0,k} = 1 - \left( \frac{k^2 - \nu^2}{k^2} \right)^{1/2}, \quad \text{and}\\
		a & > |\inf\spec \tilde L_{\kappa,k}| \,,
	\end{align*}
	which can be achieved by a suitable choice of $a$ (depending on $\nu$ and $k$).
		
	We now treat the operators $A$, $B$ and $C$ individually. Clearly, we have
	$$
	\| A \| \leq a \,.
	$$
	
	Next, we show that
	$$
	\| B \| <\infty \,.
	$$
	Indeed, by 
	
	By \cite[Lemma~IV.4, Proof of Corollary I.2]{MorozovMueller2017},
	\begin{align}
		\label{eq:comparisondirachardy}
		D_\nu^{2} \geq (1-\nu^2) \left|-i\balpha\cdot\nabla - \nu|x|^{-1} \right|^{2} \gtrsim_\nu \left(|p|-\kappa|x|^{-1}\right)^{2},
	\end{align}
	where the operator on the right side acts trivially in spin space.
	Let
	$$
	\Lambda_{\nu} :=  \one_{[0,\infty)}(D_{\nu})
	$$
	and	
	$$
	F_\nu:= \one_{[0,\infty)}(D_{\nu}) \left( D_{\nu} - 1 \right) \one_{[0,\infty)}(D_{\nu}) 
	\quad\text{in}\ L^2(\R_+:\C^4) \,.
	$$
	Thus, we have for any $a>|\inf\spec F_\nu|=1-\sqrt{1-\nu^2}$,
	\begin{align*}
		\Lambda_\nu(|p|-\kappa |x|^{-1}+a)^{2}\Lambda_\nu
		& \leq 2 \Lambda_\nu((|p|-\kappa |x|^{-1})^{2}+a^{2})\Lambda_\nu \\
		& \lesssim_\nu \Lambda_\nu(D_\nu^{2} + a^{2})\Lambda_\nu \\
		& \lesssim_{a,\nu} (F_\nu+a)^{2}.
	\end{align*}
	This inequality implies, for all $k\in\Z\setminus\{0\}$,
	$$
	\one_{[0,\infty)}(D_{\nu,k}) \left( \tilde L_{\kappa,k} +a \right)^2 \one_{[0,\infty)}(D_{\nu,k})
	\lesssim_{a,\nu} (F_{\nu,k}+a)^2 \,.
	$$
	Thus, $\|B\|\lesssim 1$, as claimed.
	
	Finally, by \eqref{eq:lbound}, we have
	\begin{align*}
		\Tr C & = \Tr_{\C^2} (\tilde L_{\kappa,k} + a)^{-2}(r,r) \\
		& = (p_{\ell^+_k}-\kappa r^{-1} + a)^{-2}(r,r) + (p_{\ell^-_k}-\kappa r^{-1} + a)^{-2}(r,r) \\
		& \lesssim_{a,\ell} 1 \,.
	\end{align*}
	This completes the proof. 
\end{proof}



\def\cprime{$'$}

\end{document}